\documentclass[runningheads]{llncs}

\usepackage{graphicx}
\usepackage{tikz}
\usepackage{amsfonts}
\usepackage{float}
\usepackage{xcolor}
\usepackage{amsmath}
\usepackage{hyperref}
\usepackage{comment}

\newcommand{\no}[1]{}

\newcommand{\floor}[1]{\lfloor #1 \rfloor}

\newcommand{\bo}{\mathcal{O}}
\newcommand{\grl}{g_{rl}}

\renewcommand{\exp}[1]{\texttt{exp}(#1)}
\newcommand{\rmq}{\texttt{rmq}}
\newcommand{\psv}{\texttt{psv}}
\newcommand{\nsv}{\texttt{nsv}}
\renewcommand{\time}{\texttt{time}}
\renewcommand{\min}{\texttt{min}}
\renewcommand{\max}{\texttt{max}}

\newif\ifshowcomments
\showcommentstrue

\usetikzlibrary{positioning, arrows}

\begin{document}

\title{Balancing Run-Length Straight-Line Programs\thanks{Funded in part by Basal Funds FB0001, Fondecyt Grant 1-200038, and two Conicyt Doctoral Scholarships, ANID, Chile.}}

\author{Gonzalo Navarro \and Francisco Olivares \and
    Cristian Urbina}
\authorrunning{Navarro et al.}

\institute{CeBiB --- Center for Biotechnology and Bioengineering \\ Departament of Computer Science, University of Chile}

\maketitle
\begin{abstract}
It was recently proved that any SLP generating a given string $w$ can be transformed in linear time into an equivalent balanced SLP of the same asymptotic size. We show that this result also holds for RLSLPs, which are SLPs extended with run-length rules of the form $A \rightarrow B^t$ for $t>2$, deriving $\exp{A} = \exp{B}^t$. An immediate consequence is the simplification of the algorithm for extracting substrings of an RLSLP-compressed string. We also show that several problems like answering RMQs and computing Karp-Rabin fingerprints on substrings can be solved in $\bo(\grl)$ space and $\bo(\log n)$ time, $\grl$ being the size of the smallest RLSLP generating the string, of length $n$. We extend the result to solving more general operations on string ranges, in $\bo(\grl)$ space and $\bo(\log n)$ applications of the operation. In general, the smallest RLSLP can be asymptotically smaller than the smallest SLP by up to an $\bo(\log n)$ factor, so our results can make a difference in terms of the space needed for computing these operations efficiently for some string families.
\keywords{Run-length straight-line programs \and Substring range problems \and Repetitive strings}
\end{abstract}

\section{Introduction}
Enormous collections of data are being generated at every second nowadays. Already storing this data is becoming a relevant and practical challenge. Compression serves to represent the data within reduced space. Still, just storing the data in compressed form is not sufficient in many cases; one also requires to construct data structures that support various queries within the compressed space. For example, {\em index} data structures support the search for short patterns in compressed strings. In areas like Bioinformatics, these collections of strings are often very repetitive \cite{Przeworski2000}, which makes traditional compressors and indexes based on Shannon's entropy unsuitable for this task \cite{NavSurvey}.

Over the years, several compressors and data structures exploiting repetitiveness have been devised. Examples of this are the Lempel-Ziv family \cite{LZ76,Kreft2010} and the run-length Burrows-Wheeler transform (BWT) \cite{BW1994,GNP18}. While compressors based on Lempel-Ziv achieve the best compression ratios, indexes based on them are not very fast and provide limited functionality. On the other hand, indexes based on the BWT can efficiently solve a variety of queries over strings, but their compression ratio is far from optimal for repetitive sequences \cite{KK2022}.

Somewhere in between of Lempel-Ziv and BWT compression is {\em grammar compression}. This approach consists in constructing a deterministic context-free grammar generating only the string to be compressed; such grammars are called straight-line programs (SLPs). Although finding the smallest SLP generating a string is NP-complete \cite{Charikar2005}, there exist several heuristics   \cite{REPAIR,SEQUITUR} and approximations \cite{Jez2015,Rytter2003} producing SLPs of small size. The popularity of SLPs probably comes from their simplicity to expose repetitive patterns on strings, which is useful to avoid redundant computation in compressed space \cite{KMV2018,ZMV2021}. This makes SLPs ideal for indexing and answering queries in compressed space \cite{Bille2015,Ganardi2021}. 

A problem that complicates such computations is that the parse tree of the grammars can be arbitrarily tall. While tasks like accessing a symbol of the string in time proportional to the parse tree height is almost trivial, achieving $\bo(\log n)$ time on general grammars requires much more sophistication \cite{Bille2015}. Recently, Ganardi et al. \cite{Ganardi2021} showed that any SLP can be balanced without incurring in an (asymptotic) increase in its size. This simplified several problems that where difficult for general SLPs, but easy if the depth of their parse tree is $\bo(\log n)$. Accessing a symbol in time $\bo(\log n)$ is nearly optimal, actually \cite{VY13}. 

An extended grammar compression mechanism are the run-length SLPs introduced by Nishimoto et al. \cite{Nishimoto2016}. An RLSLPs is an SLP extended with run-length rules of the form $A \rightarrow B^t$ for some $t > 2$, which derive $\exp{A} = \exp{B}^t$. While the size of the smallest SLP generating a string of length $n$ is always $\Omega(\log n)$, the smallest RLSLP can be of size $\bo(1)$ for some string families, which exhibit a logarithmic gap between the compression power of SLPs and RLSLPs. RLSLP have recently gained popularity for indexing. For example, all known {\em locally consistent grammars} are RLSLPs, and they have been a key component in the most recent indices for repetitive text collections. A locally consistent grammar is built through consecutive applications of a locally consistent parsing, which is a method to partition a string into non-overlapping blocks, such that equal substrings are equally parsed with the possible exception of their margins. Gagie et al. \cite{jacm19} built an index based in locally consistent grammars using $\bo(r\log\log n)$ space, with which they were able to count the $occ$ occurrences of a length-$m$ pattern in optimal time $\bo(m)$ and locate them in optimal time $\bo(m + occ)$, where $r$ is the number of runs in the BWT \cite{BW1994} of the string. Kociumaka et al. \cite{talg} also built a locally consistent grammar to index a string. Their grammar can count and locate the pattern in optimal time using $\bo(\gamma\log\frac{n}{\gamma}\log^{\epsilon} n)$ space, where $\gamma$ is the size of the smallest string attractor of the string \cite{attractors}.

In this paper we extend the results of Ganardi et al.\ to RLSLPs, that is, we show that one can always balance an RLSLP in linear time without increasing its asymptotic size. This result yields a considerable simplification to the algorithm for accessing any symbol of the string in logarithmic time \cite[Appendix~A]{talg}. It has other implications, like computing range minimum queries (RMQs) \cite{FH2011} or Karp-Rabin fingerprints \cite{KR87}, in $\bo(\log n)$ time and within $\bo(\grl)$ space. We generalize those concepts and show how to compute a wide class of semiring-like functions over substrings of an RLSLP-compressed string within $\bo(g_{rl})$ space and $\bo(\log n)$ applications of the function.

\section{Terminology}

\subsection{Strings}

Let $\Sigma$ be any finite set of {\em symbols} (an {\em alphabet}). A {\em string} $w$ is any finite tuple of elements in $\Sigma$. The {\em  length} of a string is the length of the tuple, and the {\em empty string} of length $0$ is denoted by $\varepsilon$. The set $\Sigma^*$ is formed by all the strings that can be defined over $\Sigma$. For any string $w =w_1\dots w_n$, its $i$-th symbol is denoted by $w[i]  = w_i$. Similarly, $w[i: j] = w_i\dots w_j$ with $1 \leq i \leq j \leq n$, or $\varepsilon$ if $j < i$. We also define $w[:i] = w_1\dots w_i$ and $w[i:] = w_i\dots w_n$. If $x[1: n]$ and $y[1: m]$ are strings, the concatenation operation $xy$ is defined as $xy = x_1\dots x_ny_1\dots y_m$. If $w = xyz$, then $y$ (resp. $x, z$) is a {\em substring} (resp. {\em prefix}, {\em suffix}) of $w$.

\subsection{Straight-Line Programs}

A {\em straight-line program} (SLP) is a deterministic context-free grammar generating a unique string $w$. More formally, an SLP is a context free grammar $G = (V, \Sigma, R, S)$ where $V$ is the set of variables (or non-terminals), $\Sigma$ is the set of terminal symbols (disjoint from $V$), $R \subseteq V \times (V\cup\Sigma)^*$ is the set of rules and $S$ is the initial variable; satisfying that each variable has only one rule associated, and that the variables are ordered in such a way that the starting variable is the greater of them, and any variable can only refer to other variables strictly lesser than itself or terminals, in the right-hand side of its rule. Any variable $A$ derives a unique string $\exp{A}$, and the string generated by the SLP, is the string generated by its starting variable. The size of an SLP is defined as the sum of the lengths of the right-hand side of its rules. The size of the smallest SLP generating a string is denoted by $g$, and is a relevant measure of repetitiveness. An SLP generating a non-empty string is often given in so-called Chomsky Normal Form, that is, with all its rules being of the form $A\rightarrow BC$ or $A\rightarrow a$ for $A,B,C$ variables, and $a$ a terminal symbol. 

While computing the smallest grammar is an NP-hard problem \cite{Charikar2005}, there exist several heuristic providing $\log$-approximations of the smallest SLP \cite{Jez2015,Rytter2003}. SLPs are popular as compression devices because several problems over strings can be solved efficiently using their SLP representation, without ever decompressing them. Examples of this are accessing to arbitrary positions of $w$, extracting substrings, and many other kind of queries \cite{Bille2015}. For several queries, it is convenient to have a balanced SLP, that is, an SLP whose parse tree has $\bo(\log n)$ depth. Recently, Ganardi et al. showed that any SLP can be balanced \cite{Ganardi2021}.

\subsection{Directed Acyclic Graph of an SLP}

A {\em directed acyclic graph} (DAG) is a directed multigraph $D$ without cycles (nor loops). We denote by $|D|$ the number of edges in this DAG. For our purposes, we assume that any DAG has a distinguished node $r$, satisfying that any other node can be reached from $r$, and has no incoming edges. We also assume that if a node has $k$ outgoing edges, they are numbered from $1$ to $k$.  The {\em sink nodes} of a DAG are the nodes without outgoing edges. The set of sink nodes of $D$ is denoted by $W$. We denote the number of paths from $u$ to $v$ as $\pi(u, v)$, and $\pi(u, V) = \sum_{v \in V}\pi(u, v)$ for a set $V$ of nodes. The number of paths from the root to the sink nodes is $n(D) = \pi(r, W)$.

One can interpret an SLP generating a string $w$ as a DAG $D$: There is a node for each variable in the SLP, the root node is the initial variable, terminal rules of the form $A \rightarrow a$ are the sink nodes, and a variable with rule $A \rightarrow B_1B_2\dots B_k$ has outgoing edges $(A, i, B_i)$ for $i \in [1..k]$. Note that if $D$ is a DAG representing $G$, then $n(D) = |\exp{G}| = |w|$.

\subsection{Run-Length Straight-Line Programs}

A {\em run-length straight-line program} (RLSLP) is an SLP extended with {\em run-length} rules \cite{Nishimoto2016}. An RLSLP can have rules of the form:

\begin{itemize}
    \item $A \rightarrow a$, for some terminal symbol $a$.
    \item $A \rightarrow A_1A_2\dots A_k$, for some variables $A_1, \dots, A_k$ and $k > 1$.
    \item $A \rightarrow B^t$, for some $t > 2$.
\end{itemize}

The string generated by a variable $A$ with rule $A \rightarrow B^t$ is $\exp{B}^t$. A run-length rule is considered to have size 2 (one word is needed to store the exponent). We denote by $g_{rl}$ to the size of the smallest RLSLP generating the string.
The depth of the RLSLP is the depth of its associated equivalent SLP, obtained by {\em unfolding} its run-length rules $A \rightarrow B^t$ into rules of the form $A \rightarrow BB\dots B$ of length $t$. Observe that a rule of the form $A \rightarrow A_1A_2\dots A_k$ can always be transformed into $\bo(k)$ rules of size 2, with one of them derivating the same string as $A$. Doing this for all rules can increase the depth of the RLSLP, but if $k$ is bounded by a constant, then this increase is only by a constant factor.

\section{Balancing Run-Length Straight-Line Programs}\label{sec:ballancing}

The idea utilized by Ganardi et al. to transform an SLP $G$ into an equivalent balanced SLP of size $\bo(|G|)$ \cite[Theorem~1.2]{Ganardi2021}, can be adapted to work with RLSLPs. First, we state some definitions and results proved in their work, which we need to obtain our result. 

\begin{definition}{(Ganardi et al. \cite[page 5]{Ganardi2021})}
Let $D$ be a DAG, and define the pairs $\lambda(v) = (\floor{\log_2\pi(r,v)}, \floor{\log_2\pi(v, W))})$. The {\em symmetric centroid decomposition (SC-decomposition)} of a DAG $D$ produces a set of edges between nodes with the same $\lambda$ pairs defined as $E_{scd}(D) = \{(u, i, v) \in E\,|\,\lambda(u) = \lambda(v)\}$, partitioning $D$ into disjoint paths called {\em SC-paths} (some of them possibly empty).
\end{definition}

The set $E_{scd}$ can be computed in $\bo(|D|)$ time. If $D$ is the DAG of an SLP $G$ this becomes $\bo{(|G|)}$. The following lemma justifies the name ``SC-paths''.

\begin{lemma}{(Ganardi et al. \cite[Lemma~2.1]{Ganardi2021})}\label{lemma:scd}
Let $D = (V , E)$ be a DAG. Then every node has at most one outgoing and at most one incoming edge from $E_{scd}(D)$. Furthermore, every path from the root r to a sink node contains at most $2\log_2 n(D)$ edges that do not belong to $E_{scd}(D)$.
\end{lemma}

Note that the sum of the lengths of all SC-paths is at most the number of nodes of the DAG, or the number of variables of the SLP.

The following definition and technical lemma are needed to construct the building blocks of our balanced RLSLPs.

\begin{definition}{(Ganardi et al. \cite[page~7]{Ganardi2021})}
A {\em weighted string} is a string $w \in \Sigma^*$ equipped with a {\em weight function} $||\cdot||: \Sigma \rightarrow \mathbb{N} \backslash \{0\}$, which is extended homomorphically. If $A$ is a variable in an SLP $G$, then we also write $||A||$ for the weight of the string $\exp{A}$ derived from $A$.
\end{definition}

\begin{lemma}{(Ganardi et al. \cite[Proposition~2.2]{Ganardi2021})}\label{lemma:weighted}
For every non-empty weighted string $w$ of length $n$ one can construct in linear
time an SLP $G$ with the following properties:
\begin{itemize}
    \item $G$ contains at most $3n$ variables
    \item All right-hand sides of $G$ have length at most 4
    \item $G$ contains suffix variables $S_1 , . . . , S_n$ producing all non-trivial suffixes of $w$
    \item every path from $S_i$ to some terminal symbol $a$ in the derivation tree of $G$ has length at most
$3 + 2(\log_2 ||S_i|| - \log_2 ||a||)$
\end{itemize}
\end{lemma}

We prove that any RLSLP can be balanced without asymptotically increasing its size. Our proof generalizes that of \cite[Theorem~1.2]{Ganardi2021} for SLPs.

\begin{theorem}\label{thm:balancing}Given an RLSLP $G$ generating a string $w$, it is possible to construct an equivalent balanced RLSLP $G'$ of size $\bo(|G|)$, in linear time, with only rules of the form $A\rightarrow a, A\rightarrow BC$, and $A \rightarrow B^t$, where $a$ is a terminal, $B$ and $C$ are variables, and $t > 2$.
\end{theorem}

\begin{proof}Without loss of generality, assume that $G$ has rules of length at most 2, so it is almost in Chomsky Normal Form, except because it has run-length rules. Transform the RLSLP $G$ into an SLP $H$ by unfolding its run-length rules, and then obtain the SC-decomposition $E_{scd}(D)$ of the DAG $D$ of $H$. Observe that the SC-paths of $H$ use the same variables of $G$, so it holds that the sum of the lengths of all the SC-paths of $H$ is less than the number of variables of $G$. Also, note that any variable $A$ of $G$ having a rule of the form $A\rightarrow B^t$ for some $t > 2$ is necessarily an endpoint of an SC-path in $D$, otherwise $A$ would have $t$ outgoing edges in $E_{scd}(D)$, which cannot happen.\footnote{Seen another way, $\lambda(A) \not= \lambda(B)$ because $\log_2 \pi(A,W) = \log_2 (t \cdot \pi(B,W)) > 1 + \log_2 \pi(B,W)$.} This implies that the balancing procedure of Ganardi et al. over $H$, which transforms the rules of variables that are not the endpoint of an SC-path in the DAG $D$, will not touch variables that originally were run-length in $G$.

Let $\rho=(A_0 , d_0 , A_1 ), (A_1 , d_1 , A_2 ), \dots , (A_{p-1} , d_{p-1} , A_p)$ be an SC-path of $D$. It holds that for each $A_i$ with $i \in [0..p-1]$, in the SLP $H$, its rule goes to two distinct variables, one to the left and one to the right. For each variable $A_i$, with $i \in [0..p-1]$, there is a variable $A_{i+1}'$ that is not part of the path. Let $A_1'A_2'\dots A_p'$ be the sequence of those variables. Let $L = L_1L_2\dots L_s$ be the subsequence of left variables of the previous sequence. Then construct an SLP of size $\bo(s) \leq \bo(p)$ for the sequence $L$ (seen as a string) as in Lemma \ref{lemma:weighted}, using $|\exp{L_i}|$ in $H$ as the weight function. In this SLP, any path from the suffix nonterminal $S_i$ to a variable $L_j$ has length at most $3 + 2(\log_2 ||S_i|| - \log_2 ||L_j||)$. Similarly, construct an SLP of size $\bo(t)\leq \bo(p)$ for the sequence $R = R_1R_2\dots R_t$ of right symbols in reverse order, as in Lemma \ref{lemma:weighted}, but with prefix variables $P_i$ instead of suffix variables. Each variable $A_i$, with $i \in [0..p-1]$, derives the same string as $w_{\ell}A_pw_{r}$, for some suffix $w_{\ell}$ of $L$ and some prefix $w_{r}$ of $R$. We can find rules deriving these prefixes and suffixes in the SLPs produced in the previous step, so for any variable $A_i$, we construct an equivalent rule of length at most 3. Add these equivalent rules, and the left and right SLP rules to a new RLSLP $G'$. Do this for all SC-paths. Finally, we add the original terminal variables and run-length variables of the RLSLP $G$, so $G'$ is an RLSLP equivalent to $G$. 

The SLP constructed for $L$ has all its rules of length at most 4, and $3s \leq 3p$ variables. The same happens with $R$. The other constructed rules also have length at most 3, and there are $p$ of them. Summing over all SC-paths we have $\bo(|G|)$ size. The original terminal variables and run-length variables of $G$ have rules of size at most 4, and we keep them. Thus, the RLSLP $G'$ has size $\bo(|G|)$.

Any path in the derivation tree of $G'$ is of length $\bo(\log n)$. Let $A_0,\dots,A_p$ be an SC-path. Consider a path from a variable $A_i$ to an occurrence of a variable that is in the right-hand side of $A_p$ in $G'$. Clearly this path has length at most 2. Now consider a path from $A_i$ to a variable $A_j'$ in $L$ with $i < j \leq p$. By construction this path is of the form $A_i\rightarrow S_k \rightarrow^* A_j'$ for some suffix variable $S_k$ (if the occurrence of $A_j'$ is a left symbol), and its length is at most $1 + 3 + 2(\log_2 ||S_k|| - \log_2 ||A_j'||) \leq4+2\log_2||A_i||-2\log_2||A_{j}'||$. Analogously, if $A_j'$ is a right variable, the length of the path is bounded by $1 + 3 + 2(\log_2 ||P_k|| - \log_2 ||A_j'||) \leq4+2\log_2||A_i||-2\log_2||A_{j}'||$. Finally, consider a maximal path  to a leaf in the parse tree of $G'$. Factorize it as
$$A_0 \rightarrow^* A_1 \rightarrow^* \dots \rightarrow^* A_{k}$$ 
where each $A_i$ is a variable of $H$ (and also of $G$). Paths $A_i \rightarrow^* A_{i+1}$ are like those defined in the paragraph above, satisfying that their length is bounded by $4+2\log_2||A_i||-2\log_2||A_{i+1}||$. Observe that between each $A_i$ and $A_{i+1}$, in the DAG $D$ there is almost an SC-path, except that the last edge is not in $E_{scd}$. The length of this path is at most
$$\sum_{i=0}^{k-1}(4+2\log_2||A_i||-2\log_2||A_{i+1}||) \leq k + 2\log_2||A_0|| - 2\log_2||A_k||$$
By Lemma \ref{lemma:scd}, $k \leq 2\log_2 n$, which yields the $\bo(\log n)$ upper bound. The construction time is linear, because the SLPs of Lemma \ref{lemma:weighted} are constructed in linear time in the lengths of the SC-paths (summing to $\bo(|G|)$), and $E_{scd}(D)$ can be obtained in time $\bo(|G|)$ (instead of $\bo(H)$) if we represent in the DAG $D$ the edges of a variable $A$ with rule $A \rightarrow B^t$ as a single edge extended with the power $t$. This way, when traversing the DAG from root to sinks and sinks to root to compute $\lambda$ values, it holds that $\pi(A, W) = t \cdot\pi(B,W)$, and that $\pi(r, B) = t  \cdot\pi(r, A) + c$, where $c$ are the paths from root incoming from other variables. Thus, each run-length edge must be traversed only once, not $t$ times.

To have rules of size at most two, delete rules in $G'$ of the form $A \rightarrow B$ (replacing all $A$'s by $B$'s), and note that rules of the form $A \rightarrow BCDE$ or $A \rightarrow BCD$ can be decomposed into rules of length $2$, with only a constant increase in size and depth.  \qed
\end{proof}

\section{Substring Range Operations in $\bo(\grl)$ space}

\subsection{Karp-Rabin Fingerprints}\label{sec:fingerprints}
To answer signature $\kappa(w[p:q]) = (\sum_{i=p}^{q} w[i] \cdot c^{i - p})\bmod \mu$, for a suitable integer $c$ and primer number $\mu$, we use the following identity for any $p' \in [p..q-1]$:

\begin{equation}
\kappa(w[p : q]) = \bigg(\kappa(w[p:p']) + \kappa(w[p' + 1:q])\cdot c^{p' - p}\bigg)\bmod \mu \label{eq:partition}
\end{equation}

and then it holds

$$
\begin{aligned}
\kappa(w[p:p']) &= \bigg(\kappa(w[p:q]) - \frac{c^{q - p}}{c^{q - p'}}\cdot \kappa(w[p' + 1 : q]) \bigg)\bmod \mu  \\
\kappa(w[p' + 1 : q]) &= \bigg( \frac{\kappa(w[p : q]) - \kappa(w[p : p'])}{c^{p' - p}} \bigg)\bmod \mu,
\end{aligned}
$$

which implies that, to answer $\kappa(w[p : q])$, we can compute $\kappa(w[1 : p -1])$ and $\kappa(w[1 : q])$ and then subtract one to another. For that reason, we only consider computing fingerprints of text prefixes. Then, the recursive calls to our algorithm just need to know the right boundary of a prefix, namely computing signature on the substring $\exp{A}[1 : j]$ of the string expanded by a symbol $A$ of our grammar can be expressed as $\kappa(A, j)$.

Suppose that we want to compute the signature of a prefix $w[1 : j]$ and that there is a rule $A\rightarrow BC$ such that $\exp{A} = w[1 : q]$, with $j \leq q$. If $j = |\exp{B}|$ or $j = q$, we can have stored $\kappa(\exp{A})$ and $\kappa(\exp{B})$ and answer directly the query. On the other hand, if $j < |\exp{B}|$, we can descend in the derivation tree of $B$ until we find a nonterminal $B'$ such that $|\exp{B'}| = j$ and answer the same as in previous case. Otherwise, $|\exp{B}| < j < |\exp{A}|$, then we can use Eq. \ref{eq:partition} and answer $(\kappa(\exp{B}) + \kappa(\exp{C}[1 : j - |\exp{B}|]\cdot c^{|\exp{B}|})\bmod\mu$, where $\kappa(\exp{C}[1 : j - |\exp{B}|]$ is obtained by descending in the derivation tree of $C$. Then, in addition to storing $\kappa(\exp{A})$ for every nonterminal $A$, we also need to store $c^{|\exp{A}|}\bmod\mu$ and $|\exp{A}|$. Therefore, the cost of computing fingerprints is just the depth of the derivation tree of $A$.

The same does not apply for run-length rules $A \rightarrow B^t$, because we cannot afford the space consumption of storing $c^{t'\cdot|\exp{B}|}\bmod\mu$ for every $1 \leq t' \leq t$, as this could give us a structure bigger than $\bo{(g_{rl})}$. Instead, we can treat run-length rules as regular rules $A \rightarrow B \dots B$. Then, we can use the following identity
$$
\kappa(\exp{B^{t'-1}}) = \bigg(\kappa(\exp{B})\cdot \frac{c^{|\exp{B}|\cdot (t'-1)} - 1}{c^{|\exp{B}|} - 1} \bigg)\bmod\mu.
$$
Namely, to compute $\kappa(\exp{B^{t'-1}})$ we can have previously stored $c^{|\exp{B}|}\bmod\mu$ and $(c^{|\exp{B}|} - 1)^{-1}\bmod\mu$ and then compute the exponentiation in time $\bo(\log t)$. With this, if $j \in [t'\cdot|\exp{B}| + 1 .. (t' + 1)\cdot|\exp{B}|]$ we can handle run-length rules signatures $\kappa(\exp{B^t}[1 : j])$ as
$$
\bigg(\kappa(\exp{B^{t' - 1}}) + \kappa(\exp{B}[1 : j - t'\cdot |\exp{B}|])\cdot c^{t'\cdot|\exp{B}|} \bigg)\bmod\mu,
$$
where $\kappa(\exp{B}[1 : j - t'\cdot|\exp{B}|])$ is obtained by descending in the derivation tree of $B$. We are saving space by storing our structure at the cost of increasing computation time. As we show later, this time is in fact logarithmic.

\subsubsection{A structure for Karp-Rabin signatures.}

We construct a structure over a balanced RLSLP from Theorem \ref{thm:balancing}, using some auxiliary arrays. We define an array $L[A] = |\exp{A}|$ consisting of the length of the expansion of each nonterminal $A$. For terminals $a$, we assume $L[a] = 1$. Also, we define arrays $K_1$ and $K_2$ such that, for each nonterminal $A$,
$$
\begin{aligned}
K_1[A] &= \kappa(\exp{A}), \\
K_2[A] &= c^{L[A]}\bmod \mu,
\end{aligned}
$$
with the Karp-Rabin fingerprint of the string expanded by $A$ and the last power of $c$ used in the signature multiplied by $c$, namely the first power needed for signing the second part of the string expanded by $A$. For terminals $a$ we assume $K_1[a] = a \bmod\mu$ and $K_2[a] = c\bmod\mu$.
In addition, for rules $A \rightarrow B^t$ we store 
$$
\begin{aligned}
E[A] &= (K_2[B] - 1)^{-1} \mod \mu. \\
\end{aligned}
$$
The arrays $L$, $K_j$, and $E$ add only $\bo(g_{rl})$ extra space.
With these auxiliary structures, we can compute fingerprints in $\bo(\log n)$ time.

\begin{theorem}[cf.~\cite{fingerprints,talg}]
\label{thm:KR}
    It is possible to construct and index of size $\bo(\grl)$ supporting Karp-Rabin fingerprints for prefixes of $w[1 : n]$ in $\bo(\log n)$ time.
\end{theorem}
\begin{proof}
Let $G$ be a balanced RLSLP of size $\bo(g_{rl})$ constructed as in Theorem \ref{thm:balancing}. We construct arrays $L$, $K_i$, and $E$ as shown above.
 To compute $\kappa(A, j)$, we do as follows:
    \begin{enumerate}
        \item If $j = L[A]$, return $K_1[A]$.
        \item If $A \rightarrow BC$, then:
            \begin{enumerate}
                \item If $j \leq L[B]$, return 
                $\kappa(B, j)$.
                \item If $L[B] < j$, return 
                $\big(K_1[B] + \kappa(C, j - L[B])\cdot K_2[B]\big)\bmod \mu.$
            \end{enumerate}         
        \item If $A \rightarrow B^t$ for $t > 2$, then:
            \begin{enumerate}
                \item If $j \leq L[B]$, return $\kappa(B, j)$.
                \item If $j \in [t'L[B] + 1 .. (t' + 1)L[B]]$ with $1 \leq t' \leq t$, let
                    $e = K_2[B]^{t' - 1}$ and
                    $f = (e - 1) \cdot E[A] \bmod\mu$, 
                then return 
                $$
                    \big(K_1[B]\cdot f + \kappa(B, j - t'L[B])\cdot (e\cdot K_2[B])\big)\bmod\mu.
                $$
            \end{enumerate}
    \end{enumerate}
    
    Every step of the algorithm takes $\bo(1)$ time, so the cost is the depth of the derivation tree of $G$. The only exception is case 3(b), in which we have an exponentiation. For a non-terminal $A \rightarrow B^t$, this exponentiation takes $\bo(\log t)$ time, which is $\bo(\log (|\exp{A}| / |\exp{B}|))$ time for managing every run-length rule. We show next that $\bo(\log (|\exp{A}| / |\exp{B}|))$ telescopes to $\bo(\log |\exp{A}|)$, thus we obtain $\bo(\log n)$ time for the overall algorithm time.
    
    The telescoping argument is as follows. We prove by induction that the cost $k(A)$ to compute $\kappa(A,j)$ is at most $h(A)+\log |\exp{A}|$, where $h(A)$ is the height of the parse tree of $A$ and $j$ is arbitrary. Then in case 2 we have $k(A) \le 1+\max(k(B),k(C))$, which by induction is $\le 1+\max(h(B),h(C))+\log|\exp{A}| = h(A)+\log|\exp{A}|$. In case 3 we have 
    $k(A) \le 1+\log (|\exp{A}| / |\exp{B}|)+k(B)$, and since by induction $k(B) \le h(B)+\log|\exp{B}|$, we obtain  $k(A) \le h(A)+\log|\exp{A}|$. Since $G$ is balanced, this implies $k(A)=\bo(\log n)$ when $A$ is the root symbol.\qed
\end{proof}

\subsection{Range Minimum Queries}

A {\em range minimum query} (RMQ) over a string returns the position of the leftmost occurrence of the minimum within a range. For these type of queries, we can provide an $\bo(\grl)$ space and $\bo(\log n)$ time solution.

\begin{theorem}\label{thm:RMQ}
    It is possible to construct and index of size $\bo(\grl)$ supporting RMQs in $\bo(\log n)$ time.
\end{theorem}

\begin{proof}Let $G$ be a balanced RLSLP of size $\bo(g_{rl})$ constructed as in Theorem \ref{thm:balancing}. We define $\rmq(A, i, j)$ as the pair $(a, k)$ where $a$ is the least symbol in $\exp{A}[i:j]$, and $k$ is the absolute position within $\exp{A}$ of the leftmost occurrence of $a$ in $\exp{A}[i:j]$. Store the values $L[A] = |\exp{A}|$, and $M[A] = \rmq(A, 1, L[A])$, for every variable $A$, as arrays. These arrays add only $\bo(g_{rl})$ extra space. To compute $\rmq(A,i, j)$, do as follows:
    \begin{enumerate}
        \item If $i = 1$ and $j = L[A]$, return $M[A]$.
        \item If $A \rightarrow BC$, then:
        \begin{enumerate}
            \item If $i, j \leq L[B]$, return $\rmq(B, i, j)$.
            \item If $i, j > L[B]$, let $(a, k) = \rmq(C, i-L[B], j-L[B])$. Return $(a, L[B] + k)$.
            \item If $i \leq L[B]$ and $L[B] < j$ with $j-i+1 < L[A]$, let $(a_1, k_1) = \rmq(B, i, L[B])$ and $(a_2, k_2) = \rmq(C, 1, j-L[B])$. Return $(a_1, k_1)$ if $a_1 \leq a_2$,  or $(a_2, L[B] + k_2)$ if $a_2 < a_1$.

        \end{enumerate}         
        \item If $A \rightarrow B^t$ for $t > 2$, then:
        \begin{enumerate}
            \item If $i, j \in [t'L[B] +1..(t'+1)L[B]]$, let
            $(a, k) = \rmq(B, i - t'L[B], j - t'L[B])$. Return $(a, t'L[B] + k)$
            \item If $i \in [t'L[B] +1..(t'+1)L[B]]$ and $j \in [t''L[B] +1..(t''+1)L[B]]$ for some $t' < t''$. Let $(a_l, k_l) = \rmq(B, i-t'L[B], L[B])$, $(a_r, k_r) = \rmq(B, 1, j-t''L[B])$ and
            $(a_c, k_c) = M[B]$ (only if $t'' - t'>1$). Return $(a, k)$, where $a = \min(a_l, a_r, a_c)$, and $k$ is either $t'L[B] + k_l$, $t''L[B] + k_r$, or $(t'+1)L[B] + k_c$ (only if $t'' - t'>1$), depending on which of these positions correspond to an absolute position of $a$ in $\exp{A}$, and is the leftmost of them.
        \end{enumerate}
    \end{enumerate}
We analyze the number of recursive calls of the algorithm above. For cases 2(a), 2(b) and 3(a) there is only one recursive call, over a variable which is deeper in the derivation tree of $G$. In cases 2(c) and 3(b), it could be that two recursive calls occur, but overall, this can happen only one time in the whole run of the algorithm. The reason is that when two recursive calls occur at the same depth, from that point onward, the algorithm will be computing $\rmq(\cdot)$ over suffixes or prefixes of expansions of variables. If we try to compute for example $\rmq(A, i, L[A])$, and $A$ is of the form $A \rightarrow BC$, if $i < L[B]$, the call over $B$ is again a suffix call. If $A \rightarrow B^t$ for some $t > 2$, and we want to compute $\rmq(A, i, L[A])$, we end with a recursive call over a suffix of $B$ too. Hence, there are only $\bo(\log n)$ recursive calls to $\rmq(\cdot)$. The non-recursive step takes constant time, even for run-length rules, so we obtain $\bo(\log n)$ time. \qed
\end{proof}

Other related relevant queries are {\em previous smaller value} (PSV) and {\em next smaller value} (NSV) \cite{FMNtcs09,jacm19}, defined as follows:
\begin{itemize}
    \item $\psv(i)= \max(\{j \,|\, j < i, w[j] < w[i]\}\cup \{0\})$
    \item $\nsv(i) = \min(\{j \,|\, j > i, w[j] < w[i]\}\cup \{n+1\})$
    \item $\psv'(i, d)= \max(\{j \,|\, j < i, w[j] < d\}\cup \{0\})$
    \item $\nsv'(i, d) = \min(\{j \,|\, j > i, w[j] < d\}\cup \{n+1\})$
\end{itemize}

Note that the first two queries can be computed by accessing $w[i]$ in $\bo(\log n)$ time, and then calling one of the latter two queries, respectively. We show that the latter queries can be answered in  $\bo(\grl)$ space and $\bo(\log n)$ time.

\begin{theorem}\label{thm:PSV}
    It is possible to construct and index of size $\bo(\grl)$ supporting PSV and NSV queries in $\bo(\log n)$ time.
\end{theorem}

\begin{proof}Let $G$ be a balanced RLSLP of size $\bo(g_{rl})$ constructed as in Theorem \ref{thm:balancing}. Store the values $L[A] = |\exp{A}|$ and $M[A] = \min(\{\exp{A}[i]\,|\, i \in [1.. L[A]]\})$, for every variable $A$, as arrays. These arrays add only $\bo(g_{rl})$ extra space. To compute $\psv'(A, i, d)$, do as follows:
    \begin{enumerate}
        \item If $i=1$ or $M[A] \ge d$, return $0$.
        \item If $A \rightarrow a$, return $1$.
        \item If $A \rightarrow BC$, then:
        \begin{enumerate}
            \item If $i \leq L[B]+1$, return $\psv'(B, i, d)$.
            \item If $L[B]+1 < i$, let $k = \psv'(C, i - L[B], d)$. If $k > 0$, return $L[B] + k$, otherwise, return $\psv'(B, i, d)$.
        \end{enumerate}         
        \item If $A \rightarrow B^t$ for $t > 2$, then:
        \begin{enumerate}
            \item If $i \leq L[B]+1$, return $\psv'(B, i, d)$.
            \item  If $i \in [t'L[B] +1..(t'+1)L[B]]$, let $k = \psv'(B, i - t'L[B], d)$. If $k > 0$, return $t'L[B] + k$. Otherwise, return $(t'-1)L[B] + \psv'(B,i,d)$.
        \item  If $L[A]<i$, return $(t-1)L[B]+\psv'(B,i,d)$.
        \end{enumerate}
    \end{enumerate}
    The guard in point 1 guarantees that, in the simple case where $i$ is beyond $|\exp{A}|$, at most one recursive call needs more than $\bo(1)$ time. In general, we can make two calls in case 3(b), but then the second call (inside $B$) is of the simple type from there on. The case of run-length rules is similar. Thus, we obtain $\bo(\log n)$ time. The query $\nsv'$ is handled similarly. \qed
\end{proof}

\subsection{More general functions}

More generally, we can compute a wide class of functions in $\bo(\grl)$ space and $\bo(\log n)$ applications of the function.

\begin{theorem}Let $f$ be a function from strings to a set of size $n^{\bo(1)}$, such that $f(xy) = h(f(x), f(y), |x|, |y|)$ for any strings $x$ and $y$, where $h$ is a function computable in time $\bo(\time(h))$. Let $w[1:n]$ be a string. It is possible to construct an index to compute $f(w[i:j])$ in $\bo(\grl)$ space and $\bo(\time(h) \cdot \log n)$ time.
\end{theorem}

\begin{proof}Let $G$ be a balanced RLSLP of size $g_{rl}$ constructed as in Theorem \ref{thm:balancing}. Store the values $L[A] = |\exp{A}|$ and $F[A] = f(\exp{A})$, for every variable $A$, as arrays. These arrays add only $\bo(g_{rl})$ extra space because the values in $F$ fit in $\bo(\log n)$-bit words. To compute $f(A, i, j) = f(\exp{A}[i:j])$, we do as follows:
    \begin{enumerate}
        \item If $i = 1$ and $j = L[A]$, return $F[A]$.
        \item If $A \rightarrow BC$, then:
            \begin{enumerate}
                \item If $i, j \leq L[B]$, return $f(B, i, j)$.
                \item If $i, j > L[B]$, return $f(C, i-L[B], j-L[B])$.
                \item If $i \leq L[B] < j$, return $$h(f(B, i, L[B]), f(C, 1, j-L[B]), L[B]-i+1,j-L[B]).$$
            \end{enumerate}         
        \item If $A \rightarrow B^t$ for $t > 2$, then:
            \begin{enumerate}
                \item If $i, j \in [t'L[B] +1..(t'+1)L[B]]$, return 
                $f(B, i - t'L[B], j - t'L[B])$.
                \item If $i \in [t'L[B] +1..(t'+1)L[B]]$ and $j \in [t''L[B] +1..(t''+1)L[B]]$ for some $t' < t''$, let
                \begin{align*}
                f_l &= f(B, i-t'L[B], L[B])\\
                f_r &= f(B, 1, j-t''L[B])\\
                f_c(0) &= f(\varepsilon) \\
                f_c(1) &= F[B] \\
                f_c(i) &= h(f_c(i/2),f_c(i/2),L[B]^{i/2},L[B]^{i/2}) \text{ for even }i\\
                f_c(i) &= h(f_c(1), f_c(i-1), L[B], L[B]^{i-1}) \text{ for odd }i\\
                h_l &= h(f_l, f_c(t''-t'-1), (t'+1)L[B]-i+1,(t''-t'-1)L[B])
                \end{align*}
                then return
                $$h(h_l, f_r, (t'+1)L[B]-i+1 + (t''-t'-1)L[B],j-t''L[B]+1)$$
            \end{enumerate}
    \end{enumerate}
    Just like when computing RMQs in Theorem \ref{thm:RMQ}, there is at most one call in the whole algorithm invoking two non-trivial recursive calls. 
    To estimate the cost of each recursive call, the same analysis as for Theorem~\ref{thm:KR} works, because the expansion of whole nonterminals is handled in constant time as well, and the $\bo(\log t)$ cost of the run-length rules telescopes in the same way.
    
    The precise telescoping argument is as follows. We prove by induction that the cost $c(A)$ to compute $f(A,i,L[A])$ or $f(A,1,j)$ (i.e., the cost of suffix or prefix calls) is at most $\time(h)\cdot( h(A)+\log |\exp{A}|)$, where $h(A)$ is the height of the parse tree of $A$ and $i, j$ are arbitrary. Then in case 2 we have $c(A) \le \time(h)+\max(c(B),c(C))$, which by induction is at most $\time(h)\cdot(1+ \max(h(B),h(C))+\log|\exp{A}|)=\time(h)\cdot( h(A)+\log |\exp{A}|)$. In case 3 we have that the cost is $c(A) \le \time(h)\cdot \log (|\exp{A}| / |\exp{B}|) + c(B)$, which by induction yields 
    \begin{align*}
    c(A) &\leq \time(h) \cdot (\log (|\exp{A}| / |\exp{B}|) + h(B) + \log |\exp{B}|))\\ &\le \time(h)\cdot( h(A)+\log |\exp{A}|)
     \end{align*}
     In the case that two non-trivial recursive calls are made at some point when computing $f(A_k, i, j)$, this is the unique point in the algorithm where it happens, so we charge only $\time(h)\cdot(h(A_k)+\log|\exp{A_k}|)$ to the cost of $A_k$. Then the total cost of the algorithm starting from $A_0$ is at most $\time(h)\cdot(h(A_0) + \log|\exp{A_0}|)$ plus the cost $\time(h)\cdot(h(A_k) + \log|\exp{A_k}|)$ that we did not charge to $A_k$. This at most doubles the cost, maintaining it within the same order. Because the grammar is balanced, we obtain $\bo(\time(h) \cdot \log n)$ time.\qed
\end{proof}

\section{Conclusion}

In this work, we have shown that any RLSLP can be balanced in linear time without increasing it asymptotic size. This allows us to compute several substring range queries like RMQ, PSV/NSV, and Karp-Rabin fingerprints $\bo(\log n)$ time within $\bo(\grl)$ space. More generally, in $\bo(\grl)$ space we can compute the wide class of substring functions that satisfy $f(xy) = h(f(x), f(y), |x|, |y|)$, in $\bo(\log n)$ times the cost of computing $h$. Our work also simplifies some previously established results like retrieving substrings in $\bo(\log n)$ space and within $\bo(\grl)$ space. 

An open challenge is to efficiently count the number of occurrences of a pattern in the string, within $\bo(\grl)$ space \cite[Appendix~A]{talg}; we believe this could be possible now on balanced RLSLPs.

%
%
%
\bibliographystyle{splncs04}
\bibliography{bibliography}

\end{document}